\documentclass{amsart}

\usepackage{amssymb,latexsym,amsfonts,amsmath}
\usepackage{graphicx}
\usepackage{paralist}
\usepackage{dsfont}
\usepackage{eso-pic}
\usepackage{epsfig}
\usepackage{mathtools}
\usepackage{epstopdf}

\newtheorem{theorem}{Theorem}[section]

\newtheorem{corollary}[theorem]{Corollary}

\newtheorem{definition}[theorem]{Definition}

\newtheorem{remark}[theorem]{Remark}
\numberwithin{equation}{section}

\newcommand{\R}{{\mathbb{R}}}
\newcommand{\F}{{\mathbb{F}}}

\newcommand{\q}{\mathsf{q}}
\newcommand{\p}{\mathsf{p}}

\DeclareMathOperator{\Tr}{Tr}
\DeclareMathOperator{\diff}{d}

\newcommand{\sigalg}{\mathcal{F}}

\newcommand{\ul}{\underline}
\newcommand{\ol}{\overline}
\newcommand{\Let}{:=}
\newcommand{\EE}{\mathds{E}}
\newcommand{\PP}{\mathds{P}}

\begin{document}

\begin{abstract}
	Incremental stability is a property of dynamical systems ensuring the uniform asymptotic stability of each trajectory rather than a fixed equilibrium point or trajectory. Here, we introduce a notion of incremental stability for stochastic control systems and provide its description in terms of existence of a notion of so-called incremental Lyapunov functions. Moreover, we provide a backstepping controller design scheme providing controllers along with corresponding incremental Lyapunov functions rendering a class of stochastic control systems, namely, stochastic Hamiltonian systems with jumps, incrementally stable. To illustrate the effectiveness of the proposed approach, we design a controller making a spring pendulum system in a noisy environment incrementally stable.
\end{abstract}

\title[Backstepping Design for Incremental Stability of Stochastic Hamiltonian Systems with Jumps]{Backstepping Design for Incremental Stability of Stochastic Hamiltonian Systems with Jumps}
\author[P. Jagtap]{Pushpak Jagtap} 
\author[M. Zamani]{Majid Zamani} 

\address{$^1$Department of Electrical and Computer Engineering, Technical University of Munich, D-80290 Munich, Germany.}
\email{\{pushpak.jagtap,zamani\}@tum.de}
\urladdr{http://www.hcs.ei.tum.de}
\maketitle
\section{Introduction}
The notion of incremental stability focuses on the convergence of trajectories with respect to each other rather than with respect to an equilibrium point or a fixed trajectory. This notion of stability has gained significant attention in recent years due to its potential applications in the study of nonlinear systems. Examples of such applications include synchronization of cyclic feedback systems \cite{hamadeh2012global}; construction of symbolic models \cite{pola2008approximately}, \cite{majumdar2012approximately}; modeling of nonlinear analog circuits \cite{bond2010compact}; and synchronization of interconnected oscillators \cite{stan2007analysis}.   

In the past few years, there have been several results characterizing the notion of incremental stability for non-probabilistic dynamical systems using notions of so-called incremental Lyapunov functions and contraction metric. The interested readers may consult the results in \cite[and references therein]{angeli2002lyapunov,pavlov2006uniform,LOHMILLER1998683,zamani2011backstepping,zamani2013backstepping} for more detailed information about different characterizations of incremental stability. Furthermore, there have been several results on the construction of state feedback controllers enforcing a class of non-probabilistic control systems incrementally stable. Examples include results on smooth strict-feedback form systems \cite{zamani2011backstepping} and a class of (not-necessarily smooth) control systems \cite{zamani2013backstepping}.

In recent years, similar notions of incremental stability have been introduced for different classes of stochastic systems including stochastic control systems \cite{zamani2014symbolic}, stochastic switched systems \cite{zamani2015symbolic}, randomly switched stochastic systems \cite{zamani2014approximately}, and their descriptions using some notions of incremental Lyapunov functions. In addition, there have been several results in the literature studying incremental stability of stochastic systems using a notion of contraction metric. Examples include the results on stochastic dynamical systems \cite{pham2009contraction} and a class of stochastic hybrid systems \cite{zhang2013incremental}.

There exists a plethora of results for designing state feedback controllers enforcing some classes of stochastic systems stable with respect to an equilibrium point or a fixed trajectory. Examples include the results based on backstepping and inverse optimality \cite{deng1999output}, on strict-feedback form stochastic systems \cite{khoo2013finite}, based on passivity for stochastic port-Hamiltonian systems \cite{6361273}, on a backstepping approach for stochastic Hamiltonian systems \cite{wu2012backstepping}, 
on input-to-state stability of stochastic retarded systems \cite{huang2009input}, 
and finally on stabilization of jump stochastic systems \cite{lin2012dissipative}. 
Unfortunately, to the best of our knowledge, there is no work available in the literature on the synthesis of state feedback controllers rendering a class of nonlinear stochastic systems incrementally stable. This is unfortunate because, based on our motivation from symbolic control, incremental stability is a key property enabling the construction of bisimilar finite abstractions of continuous-time stochastic systems \cite{zamani2014symbolic,zamani2015symbolic,zamani2014approximately,zamani2016} and, hence, automated controller synthesis methodologies for this class of systems.

The main objective of this work is to propose, for the first time, a state feedback design scheme providing controllers enforcing a class of stochastic systems incrementally stable. The paper is divided in two major parts. In the first part, we introduce a coordinate invariant notion of incremental stability for stochastic control systems with jumps and provide its description in terms of existence of a notion of so-called incremental Lyapunov functions. In the second part, we provide a feedback controller design approach based on backstepping scheme providing controllers together with the corresponding incremental Lyapunov functions enforcing a class of stochastic control systems, namely, stochastic Hamiltonian systems with jumps, incrementally stable. Further, we illustrate the effectiveness of the proposed results by designing a state feedback controller making a spring pendulum subject to stochastically vibrating ceiling with jumps incrementally stable.


\section{Stochastic Control Systems}\label{II}
\subsection{Notations}
The symbols $ \R$, $\R^+,$ and $\R_0^+ $ denote the set of real, positive, and non-negative real numbers, respectively. We use $ \R^{n\times m} $ to denote a vector space of real matrices with $ n $ rows and $ m $ columns. The identity matrix in $  \R^{n\times n} $ is denoted by $ I_n $ and zero matrix in $  \R^{n\times m} $ is denoted by $ 0_{n\times m} $. The $ e_i\in \R^n $ denotes the vector whose all elements are zero, except the $ i^{th} $ element, which is one. Given a vector \mbox{$x\in\mathbb{R}^{n}$}, we denote by $x_{i}$ the $i$-th element of $x$, and by $\Vert x\Vert$ the Euclidean norm of $x$. Given a matrix $ A \in \R^{n\times m}$, $ A^T $ represents transpose of matrix $ A $ and $ \|A\|_F $ represents the Frobenius norm of $ A $ defined as $ \|A\|_F=\sqrt{\mathsf{Tr}(AA^T)} $, where $ \mathsf{Tr}(\cdot) $ denotes the trace of a square matrix. For all $ x_i\in\R^{n_i} $, $ [x_1; x_2;\ldots; x_N] $ represents a vector in $ \R^n $ where $n=\sum_{i=1}^Nn_i$. The symbol $ A\otimes B $ represents the Kronecker product of matrices $ A $ and $ B $. The diagonal set $ \bigtriangleup \subset \R^{2n} $ is defined as $ \bigtriangleup=\{(x,x)|x\in \R^n\} $. A continuous function $ \gamma  :\R^+_0 \rightarrow \R^+_0 $ belongs to class $ \mathcal{K} $ if it is strictly increasing and $ \gamma(0)=0 $; it belongs to class $ \mathcal{K}_\infty $ if $ \gamma \in \mathcal{K}  $ and $ \gamma(r) \rightarrow \infty $ as $ r \rightarrow \infty $. A continuous function $ \beta:\R^+_0 \times \R^+_0 \rightarrow \R^+_0 $ belongs to class $\mathcal{KL}$ if for each fixed $ s $, the map $ \beta(r,s) $ belongs to $ \mathcal{K} $ with respect to $ r $ and, for each fixed $ r\neq0 $, the map $ \beta(r,s) $ is decreasing with respect to $ s $ and $ \beta(r,s) \rightarrow 0 $ as $ s\rightarrow \infty $. For any $x, y, z \in \R^n $, a function $\textbf{d}:\R^n \times \R^n \rightarrow \R^+_0$ is a metric on $ \R^n $ if the following conditions hold: (i) $ \textbf{d}(x,y)=0 $ if and only if $x=y$; (ii) $ \textbf{d}(x,y)= \textbf{d}(y,x) $; and (iii) $\textbf{d}(x,z)\leq \textbf{d}(x,y)+\textbf{d}(y,z)$. Given a measurable function $f: \R^+_0 \rightarrow \R^n $, the (essential) supremum of $ f $ is denoted by $ \|f\|_\infty $; we recall that $ \|f\|_\infty$ := (ess)sup$\{\|f(t)\|,t\geq0\}$.
\subsection{Stochastic control systems}
Let the triplet $(\Omega, \mathcal{F}, \mathds{P})$ denote a probability space with a sample space $ \Omega $, filtration $ \mathcal{F} $, and the probability measure $ \mathds{P} $. The filtration $\mathds{F}= (\mathcal{F}_s)_{s\geq 0}$ satisfies the usual conditions of right continuity and completeness \cite{oksendal}. Let $ (W_s)_{s\geq0} $ be an $ r $-dimensional $ \mathds{F} $-Brownian motion and $ (P_s)_{s\geq0} $ be an $ \tilde{r} $-dimensional $ \F $-Poison process. We assume that the Poisson process and the Brownian motion are independent of each other. The Poisson process $ P_s:=[P_s^1;\ldots; P_s^{\tilde{r}}] $ models $ \tilde{r} $ kinds of events whose occurrences are assumed to be independent of each other.
\begin{definition} \label{definition1}
	A stochastic control system is a tuple $\textstyle \Sigma_s=(\R^n, \R^m, \mathcal{U}, f, \sigma, \rho)$, where:
	\begin{compactitem}
		\item $\R ^n$ is the state space;
		\item $\R^m$ is the input space;
		\item $ \mathcal{U} $ is a subset of the set of all $\mathds{F}$-progressively measurable processes with values in $\R^m$; see \cite[Def. 1.11]{ref:KarShr-91};
		\item $f : \R ^n \times \R^m \rightarrow\R ^n $ satisfies the following Lipschitz assumption: there exist constants $L_x,L_u\in\R^+$ such that: $ \| f(x,u)-f(x',u') \| \leq L_x \| x-x'\|+L_u \| u-u'\|$ $ \forall  x, x'\in \R ^n$ and $ \forall u, u'\in \R^m $;
		\item $\sigma :\R^n \rightarrow \R^{n\times r}$ satisfies the following Lipschitz assumption: there exists a constant $ L_\sigma \in \R^+_0 $ such that: $ \| \sigma(x)-\sigma(x') \| \leq L_\sigma \| x-x' \|$ $\forall  x, x'\in \R ^n$;
		\item $\rho :\R^n \rightarrow \R^{n\times \tilde{r} }$ satisfies the following Lipschitz assumption: there exists a constant $ L_\rho \in \R^+_0 $ such that: $ \| \rho(x)-\rho(x') \| \leq L_\rho \| x-x' \|$ $\forall  x, x'\in \R ^n$.
	\end{compactitem}
\end{definition}
A stochastic process $\xi : \Omega \times \R^+_0\rightarrow \R^n$ is said to be a \textit{solution process} of $\Sigma_s$ if there exists $ \upsilon\in \mathcal{U} $ satisfying
\begin{equation}
	\diff \xi = f(\xi, \upsilon)\diff t + \sigma(\xi)\diff W_t+\rho(\xi)\diff P_t,
	\label{stochastic_Process}
\end{equation}
$ \mathds{P} $-almost surely ($ \mathds{P} $-a.s.), where $ f(\cdot) $, $ \sigma(\cdot) $, and $\rho(\cdot)$ are the drift, diffusion, and reset terms, respectively. We emphasize that postulated assumptions on $f$, $\sigma$, and $\rho$ ensure the existence and uniqueness of the solution process \cite{Oks-jump}. Throughout the paper, we use the notation $ \xi_{a\upsilon}(t) $ to denote the value of a solution process at time $ t \in \R^+_0 $ under the input signal $ \upsilon $ and with initial condition $ \xi_{a\upsilon}(0) = a $ $ \mathds{P} $-a.s., in which $ a $ is a random variable that is measurable in $ \mathcal{F}_0 $. Here, we assume that the Poisson processes $P_s^i$, for any $i\in\{1,\ldots,\tilde{r}\}$, have the constant rates of $ \lambda_i $.  
\subsection{Incremental stability for stochastic control systems}
This subsection introduces a coordinate invariant notion of incremental input-to-state stability for stochastic control systems. The stability notion discussed here is the generalization of the ones defined in \cite{zamani2013backstepping}, \cite{zamani2011backstepping} for non-probabilistic control systems. 
\begin{definition} \label{definition2}
	A stochastic control system $\textstyle \Sigma_s$ is incrementally input-to-state stable ($ \delta_\exists $-ISS-M$_k $) in the $ k^{th} $ moment, where $ k\geq1 $, if there exist a metric $\textbf{d}$, a $ \mathcal{KL} $ function $ \beta $, and a $ \mathcal{K}_\infty $ function $ \gamma $ such that for any $ t \in \R^+_0 $, any $ \R^n $-valued random variables $ a $ and $ a' $ that are measurable in $ \mathcal{F}_0 $, and any $\upsilon,\upsilon'\in\mathcal{U}$, the following condition is satisfied:
	\begin{equation}
		\begin{split}
			\EE[(\textbf{d}&(\xi_{a\upsilon}(t),\xi_{a'\upsilon'}(t)))^k]\leq\beta(\EE[(\textbf{d}(a,a'))^k],t)+\gamma(\EE[\|\upsilon-\upsilon'\|^k_\infty]).
		\end{split}
		\label{nnn}
	\end{equation} 
\end{definition}
\begin{remark}
		Note that if one uses the natural Euclidean metric rather than a general metric $\textbf{d}$ in Definition \ref{definition2}, the notion reduces to the one defined in \cite[Definition 3.1]{zamani2014symbolic} which is not invariant under changes of coordinates. Observer that changes of coordinates are one of the main tools used in the backstepping design scheme including the one proposed in this paper.
	\end{remark}
	One can describe $ \delta_\exists$-ISS-M$_k $ in terms of existence of $ \delta_\exists$-ISS-M$_k $ Lyapunov functions as defined next.
\begin{definition}\label{definition3}
	Consider a stochastic control system $\textstyle \Sigma_s $ and a continuous function $ V : \R^n \times \R^n\rightarrow\R^+_0 $ that is twice continuously differentiable on $  {\R^n \times \R^n} \setminus \bigtriangleup $. The function $ V $ is called a $ \delta_\exists$-ISS-M$_k $ Lyapunov function for $\textstyle \Sigma_s $, if it has polynomial growth rate and there exist a metric $\textbf{d}$, $ \mathcal{K}_\infty $ functions  $ \underline{\alpha}, \overline{\alpha}$, $ \varphi $, and a constant $ \kappa\in \R^+ $, such that:
	\begin{enumerate}
		\item[(i)] $ \underline{\alpha} $(resp. $ \overline{\alpha} $ and $ \varphi $) is a convex (resp. concave) function;
		\item[(ii)] $  \forall x, x'\in\R^n$, $\underline{\alpha}((\textbf{d}(x,x'))^k) \leq V(x, x')\leq\overline{\alpha}((\textbf{d}(x,x'))^k) $; 
		\item[(iii)] $  \forall x, x'\in\R^n, x\neq x', $ and $  \forall u, u'\in \R^m$,
	\end{enumerate}
	\begin{equation}
		\begin{split}
			\mathcal{L}V(x,x'):=&\begin{bmatrix}\partial_x V(x,x') & \partial_{x'} V(x,x')\end{bmatrix}\begin{bmatrix}
				f(x,u)\\ 
				f(x',u')
			\end{bmatrix}+ \frac{1}{2}\mathsf{Tr}\bigg(\begin{bmatrix}
				\sigma(x)\\ 
				\sigma(x')
			\end{bmatrix} \begin{bmatrix}
			\sigma^T(x) &\hspace{-.5em} \sigma^T(x')
		\end{bmatrix}\begin{bmatrix}
		\partial_{x,x}V &\hspace{-.5em} \partial_{x,x'}V \\ 
		\partial_{x',x}V &\hspace{-.5em} \partial_{x',x'}V
	\end{bmatrix}\bigg)\\
	&+\sum_{i=1}^{\tilde{r}}\lambda_i\big(V(x+\rho(x)e_i,x'+\rho(x')e_i)-V(x,x')\big) \\ \leq& -\kappa V(x,x')+\varphi(\|u -u'\|^k),\nonumber
\end{split}
\end{equation}
\end{definition}
where $ \mathcal{L} $ is the infinitesimal generator of the stochastic process $\xi$ in (\ref{stochastic_Process}) acting on function $V$ \cite[equation (23)]{julius1} and the symbols $ \partial_{x} $ and $ \partial_{x,x'} $ represents first and second-order partial derivatives with respect to $ x $ and $ x' $, respectively. The following theorem describes $ \delta_\exists$-ISS-M$_k $ in terms of existence of $ \delta_\exists$-ISS-M$_k $ Lyapunov functions. \\ 

\begin{theorem}\label{theorem1}
	A stochastic control system $\textstyle \Sigma_s $ is $ \delta_\exists$-ISS-M$_k $ if it admits a $ \delta_\exists$-ISS-M$_k $ Lyapunov function.
\end{theorem}
The proof is similar to the one of Theorem 3.3 in \cite{zamani2014symbolic} and is omitted due to lack of space.
\begin{proof}
For any time instance $t \ge 0$, any $\upsilon,\upsilon'\in\mathcal{U}$, and any random variable $a$ and $a'$ that are $\sigalg_0$-measurable, one obtains
	\begin{align}
		\notag \EE \left[ V(\xi_{a\upsilon}(t),\xi_{a'\upsilon'}(t)) \right] &= \EE \Big[V\big(\xi_{a\upsilon}(0),\xi_{a'\upsilon'}(0)\big) + \int_{0}^{t} \mathcal{L} V(\xi_{a\upsilon}(s),\xi_{a'\upsilon'}(s))\diff s\Big]
		\\ \notag & \le \EE \left[ V\big(\xi_{a\upsilon}(0),\xi_{a'\upsilon'}(0)\big)\right] +\EE\Big[\int_{0}^{t} \big(-\kappa V\big(\xi_{a\upsilon}(s),\xi_{a'\upsilon'}(s)\big) +\varphi(\Vert \upsilon(s)-\upsilon'(s)\Vert^k)\big) \diff s \Big]
		\\ \notag & \le \EE \left[ V\big(\xi_{a\upsilon}(0),\xi_{a'\upsilon'}(0)\big)\right] +\int_{0}^{t} \Big(-\kappa\EE\left[ V\big(\xi_{a\upsilon}(s),\xi_{a'\upsilon'}(s)\big)\right]+\EE\Big[\varphi(\Vert \upsilon-\upsilon'\Vert^k_\infty)\Big]\Big)\diff s,
	\end{align}
	where the first equality is an application of the It\^{o}'s formula for jump diffusions thanks to the polynomial rate of the function $V$ \cite[Theorem 1.24]{Oks-jump} and the first inequality is because of condition iii) in Definition \ref{definition3}.
By virtue of Gronwall's inequality, one obtains
\begin{align}
\label{lyap_bnd} \EE[ V(\xi_{a\upsilon}(t),\xi_{a'\upsilon'}(t)) ]\le \EE[V(a,a')]\mathsf{e}^{-\kappa t}  + \frac{1}{\kappa}\EE\big[\varphi(\Vert \upsilon-\upsilon'\Vert^k_\infty)\leq\EE[V(a,a')]\mathsf{e}^{-\kappa t}  + \frac{1}{\kappa}\varphi\big(\EE\big[\Vert \upsilon-\upsilon'\Vert^k_\infty\big]\big),
\end{align}
where the last inequality follows from Jensen's inequality due to the concavity assumption on the function $\varphi$ \cite[p. 310]{oksendal}.
In view of Jensen's inequality, inequality \eqref{lyap_bnd}, the convexity of $\ul\alpha$, the concavity of $\ol\alpha$, and condition ii) in Definition \ref{definition3}, we have the following chain of inequalities	
	\begin{align}\notag
		{\ul\alpha}\big(\EE \big[ (\mathbf{d}(\xi_{a\upsilon}(t)&,\xi_{a'\upsilon'}(t)))^k\big]\big)  \le \EE \big[ {\ul\alpha}\big((\textbf{d}( \xi_{a\upsilon}(t),\xi_{a'\upsilon'}(t)))^k\big) \big] \le\EE \left[ V(\xi_{a\upsilon}(t),\xi_{a'\upsilon'}(t))\right] \\\notag&\le \EE[V(a, a')]\mathsf{e}^{-\kappa t}  + \frac{1}{\kappa}\varphi\big(\EE\big[\Vert \upsilon-\upsilon'\Vert^k_\infty\big]\big)\leq\EE\big[\ol\alpha\big((\textbf{d}( a, a'))^k\big)\big]\mathsf{e}^{-\kappa t}  + \frac{1}{\kappa}\varphi\big(\EE\big[\Vert \upsilon-\upsilon'\Vert^k_\infty\big]\big)
		\\ \notag &\le \ol\alpha\big(\EE\big[(\textbf{d}( a, a'))^k\big]\big)\mathsf{e}^{-\kappa t}  + \frac{1}{\kappa}\varphi\big(\EE\big[\Vert \upsilon-\upsilon'\Vert^k_\infty\big]\big),
	\end{align}
which in conjunction with the fact that $\ul\alpha \in \mathcal{K}_{\infty}$ leads to
	\begin{align*}
	\EE  \big[(\textbf{d} (\xi_{a\upsilon}(t),\xi_{a'\upsilon'}(t)))^k\big]&\le{\ul\alpha}^{-1} \Big(\ol\alpha\big(\EE\big[(\textbf{d}( a, a'))^k\big]\big)\mathsf{e}^{-\kappa t}  + \frac{1}{\kappa}\varphi\big(\EE\big[\Vert \upsilon-\upsilon'\Vert^k_\infty\big]\big)\Big) 
	\\ & \le \ul\alpha^{-1} \big(2\ol\alpha\big(\EE\big[(\textbf{d}( a, a'))^k\big]\big)\mathsf{e}^{-\kappa t} \big) + \ul\alpha^{-1}\Big(\frac{2}{\kappa}\varphi\big(\EE\big[\Vert \upsilon-\upsilon'\Vert^k_\infty\big]\big)\Big).
	\end{align*}
	Therefore, by introducing functions $\beta$ and $\gamma$ as
	\begin{align}\label{com-fn}
		\beta(y,t) \Let {\ul\alpha}^{-1}\big(2\ol\alpha(y)\mathsf{e}^{-\kappa t}\big),~~ \gamma(y) \Let {\ul\alpha}^{-1}\Big(\frac{2}{\kappa}\varphi(y)\Big),
	\end{align}
	for any $y,t\in\R^+_0 $, inequality (\ref{nnn}) is satisfied. Note that if $\ul\alpha^{-1}$ satisfies the triangle inequality (i.e., $\ul\alpha^{-1}(a+b)\leq\ul\alpha^{-1}(a)+\ul\alpha^{-1}(b)$), one can remove the coefficients $2$ in the expressions of $\beta$ and $\gamma$ in \eqref{com-fn} to get a less conservative upper bound in \eqref{nnn}.
\end{proof}

\section{Backstepping Design Procedure}\label{III}
This section contains the main contribution of the paper. Here, we propose a backstepping control design scheme for a class of stochastic control systems, namely,  stochastic Hamiltonian systems with jumps. The proposed methodology provides controllers rendering the closed loop system $ \delta_\exists$-ISS-M$_k $.
A stochastic Hamiltonian system with jumps is a stochastic control system $\Sigma=(\R^{2n}, \R^n, \mathcal{U}, f, \sigma, \rho)$ described by stochastic differential equations  
\begin{equation}
	\textstyle \Sigma :\left\{\begin{split}
		\diff \q=& \partial_p H (\q,\p)\diff t,\\
		\diff \p=& \Big(-\partial_q H(\q,\p)+ b(\q,\p)+G(\q)\upsilon \Big)\diff t+\sigma(\q)\diff  W_t+\rho(\q)\diff P_t,
		\label{ham1}
	\end{split}\right.
\end{equation}
where $q= \q(\omega,t)\in \mathbb{R}^n $, $ \forall t \in \mathbb{R}_0^+$ and $\forall \omega\in\Omega$, is a generalized coordinate vector of $ n $-degree-of-freedom system; $ p= \p(\omega,t) \in \mathbb{R}^n $, $ \forall t \in \mathbb{R}_0^+$ and $\forall \omega\in\Omega$, represents a vector of generalized momenta and defined as $ \p\diff t=M(\q)\diff \q $, where $ M(q) $ is a symmetric, nonsingular, and positive definite inertia matrix; $ b(q,p) $ is a smooth damping term; $ G(q)\upsilon $ is the control force caused by $ G(q) $, a nonsingular smooth square matrix, and by control input $ \upsilon$ acting on the system; $\sigma(q)$ is the diffusion term; $\rho(q)$ is the reset term capturing the magnitude of jumps; and $ \partial_qH$ and $ \partial_p H$ represent first order partial derivative of function $ H$ with respect to $ q $ and $ p $, respectively, where $ H$ is a continuously differentiable Hamiltonian function represented in terms of total energy of the system as the following
\begin{equation}
	H(q,p)=\frac{1}{2}p^TM^{-1}(q)p+\Xi(q),
	\label{ham_fn}
\end{equation}  
where $ \Xi(q) $ represents potential energy of the system. By substituting (\ref{ham_fn}) into (\ref{ham1}), the dynamics of $ \Sigma $ can be rewritten as
\begin{equation}
	\textstyle \Sigma:\left\{\begin{split}
		\diff \q=& M^{-1}(\q)\p\diff t,\\
		\diff \p=& \Big(\eta(\q,\p)+G(\q)\upsilon\Big)\diff t+\sigma(\q)\diff W_t+\rho(\q)\diff P_t,
	\end{split}\right.
	\label{ham2}
\end{equation} 
where $\eta(\q,\p)=-\partial_q H(\q,\p)+ b(\q,\p).$ 
\begin{remark}
		Note that the dynamic considered in (\ref{ham2}) is the generalization of the ones given in \cite{wu2012backstepping} and \cite{iwankiewicz2016dynamic}. It extends the former by including the jump term and the latter by adding the diffusion term and representing more general stochastic Hamiltonian systems.
	\end{remark}
	As we already emphasized after Definition \ref{definition1}, in order to ensure the existence and uniqueness of the solution process of $ \Sigma $ in (\ref{ham2}), one requires a Lipschitz assumption on the drift term which implies:
		\begin{equation}
			\|M^{-1}(q)p-M^{-1}(q')p'\|\leq L_1\|q-q'\|+L_2\|p-p'\|,
			\label{condition}
		\end{equation}
		for some $ L_1,L_2 \in\R^+$ and any $q,q',p,p'\in\R^n$.
	
	We can now state the main result of the paper on the backstepping controller design scheme providing controllers rendering the considered class of stochastic control systems $ \delta_\exists$-ISS-M$_k$ for any $ k\geq2 $.
	
	\begin{theorem} \label{theorem2}
		Consider the stochastic control system $\Sigma$ of the form (\ref{ham2}). The state feedback control law 
		\begin{align}\notag
			\upsilon= &G^{-1}(\q)\Big(-\eta(\q,\p)+\kappa_1 \frac{\diff M(\q)}{\diff t} \q-\kappa_1^2M(\q)\q\\\label{thmeq}
			&-\Big(\lambda\frac{(2^{k-1}-1)}{k}+\frac{L_2}{s_1\varepsilon_1^{s_1}}+\frac{\min\{n,r\}L_\sigma^2(k-1)}{2s_2\varepsilon_2^{s_2}}\Big)\big(\p+\kappa_1M(\q)\q\big)+\hat{\upsilon}\Big),
		\end{align}
		renders the closed-loop stochastic control system $ \Sigma $ $\delta_\exists$-ISS-M$_k$ for $ k>2 $ with respect to the input $ \hat{\upsilon} $, for all 
		\[ \kappa_1> L_1+\frac{\max\{L_2,1\} \varepsilon _1^{r_1}}{r_1}+\frac{\min\{n,r\}L_\sigma^2\varepsilon _2^{r_2}(k-1)}{2r_2}+\frac{2^{k-1}L_\rho^k \lambda}{k} ,\] 
		where $r_1=\frac{k}{k-1}$, $s_1=k$, $r_2=\frac{k}{k-2}$, $s_2=\frac{k}{2}$, $ \varepsilon_1$ and $\varepsilon_2 $ are positive constants which can be chosen arbitrarily, $\lambda=\sum_{i=1}^{\tilde{r}}\lambda_i$, and $ L_1, L_2, L_\sigma$, and $L_\rho $ are the Lipschitz constants introduced in (\ref{condition}) and Definition \ref{definition1}, respectively.
	\end{theorem}
	
	Note that the term $ \frac{\diff M(\q)}{\diff t} $ in the control law (\ref{thmeq}) can be computed by using the definition of derivative of matrix \cite{wu2012backstepping} as 
	\[\frac{\diff M(\q)}{\diff t}=\frac{\partial M(\q)}{\partial q^T}\times\Big(\frac{\diff \q}{\diff t} \otimes I_n\Big)=\frac{\partial M(\q)}{\partial q^T}\times\big(M^{-1}(\q)\p\otimes I_n\big),\] where $\frac{\partial M(\q)}{\partial q^T}:=\left[\frac{\partial M(\q)}{\partial q_1}~\frac{\partial M(\q)}{\partial q_2}\cdots\frac{\partial M(\q)}{\partial q_n}\right]_{n\times n^2}$.
	\begin{proof}
	Consider a coordinate transformation as
	\begin{equation}
		\zeta=\psi(\xi)=\begin{bmatrix}
			\zeta_1\\ 
			\zeta_2
		\end{bmatrix}=\begin{bmatrix}
		\q\\ 
		\p-\alpha(\q)
	\end{bmatrix},
	\label{cordi_tran}
\end{equation}
where $ \xi=[\q^T~ \p^T]^T $ and $ \alpha(\q) =-\kappa_1M(\q)\q$ for some $ \kappa_1>0 $. The dynamics of the stochastic control system $ \textstyle \Sigma $ in (\ref{ham2}) after the change of coordinates can be written by using Ito's differentiation \cite{oksendal} as
\begin{equation}
	\textstyle \hat{\Sigma}:\left\{\begin{split}
		\diff\zeta_1=& M^{-1}(\zeta_1)(\zeta_2+\alpha(\zeta_1))\diff t,\\
		\diff\zeta_2=& \Big(\eta(\zeta_1,\zeta_2+\alpha(\zeta_1))+G(\zeta_1)\upsilon+\kappa_1  \frac{\diff M(\zeta_1)}{\diff t}\zeta_1+\kappa_1\big(\zeta_2+\alpha(\zeta_1)\big)\Big)\diff t\\&+\sigma(\zeta_1)\diff W_t+\rho(\zeta_1)\diff P_t.
	\end{split}\right.
	\label{syst1}
\end{equation}
Now consider a candidate Lyapunov function $ V_1(z_1,z_1') $, $ \forall z_1,z_1' \in \mathbb{R}^n $, for the $ \zeta_1 $-subsystem as follows
\begin{equation}\nonumber
	V_1(z_1,z_1')=\frac{1}{k}\Big((z_1-z_1')^T(z_1-z_1')\Big)^{\frac{k}{2}}.
	\label{Lyap_1}
\end{equation}   
The corresponding infinitesimal generator along $ \zeta_1 $-subsystem is given by
\begin{equation}
	\begin{split}
		\mathcal{L}V_1(z_1,z_1')=(z_1-&z_1')^T\Big((z_1-z_1')^T(z_1-z_1')\Big)^{\frac{k}{2}-1}\\
		&\Big[\big(M^{-1}(z_1)z_2-M^{-1}(z_1')z_2'\big)+\big(M^{-1}(z_1)\alpha(z_1)-M^{-1}(z_1')\alpha(z_1')\big)\Big].
		\nonumber
	\end{split}
	\label{L1}
\end{equation}
Now by using the definition of $ \alpha(z_1) $, consistency of norm, and (\ref{condition}), the infinitesimal generator reduces to
\begin{equation}
	\begin{split}
		\mathcal{L}V_1(z_1,z_1')\leq &(L_1-\kappa_1)\Big((z_1-z_1')^T(z_1-z_1')\Big)^{\frac{k}{2}}+L_2 \big((z_1-z_1')^T(z_1-z_1')\big)^{\frac{k}{2}-1}\|z_1-z_1'\| \|z_2-z_2'\|.
	\end{split}
	\label{LLL}
\end{equation}
To handle the second term, we use Young's inequality \cite{krstic1995nonlinear} as 
\begin{equation}
	ab\leq\frac{\varepsilon^r}{r}|a|^r+\frac{1}{s\varepsilon^s}|b|^s,
	\label{young_inequality}
\end{equation}
where $ \varepsilon>0 $, constants $ r,s>1 $ satisfying condition $ (r-1)(s-1)=1 $, and $a,b\in \R$. Now by using the consistency of norms and applying Young's inequality (\ref{young_inequality}), we can reduce the second term in (\ref{LLL}) to 
\begin{align}\notag
	L_2&\big((z_1-z_1')^T(z_1-z_1')\big)^{\frac{k}{2}-1}\|z_1-z_1'\|\|z_2-z_2'\|= L_2\|z_1- z_1'\|^{k-1}\|z_2-z_2'\|\\\label{three}
	&\quad\quad\leq\frac{L_2 \varepsilon _1^{r_1}}{r_1}\Big(( z_1-  z_1')^T(z_1-z_1')\Big)^{\frac{k}{2}}+\frac{L_2}{s_1\varepsilon_1^{s_1}}\Big((z_2-z_2')^T(z_2-z_2')\Big)^{\frac{k}{2}},
\end{align}
where,$r_1=\frac{k}{k-1}$, $s_1=k$, and $ \varepsilon _1$ is any positive constant. After substituting inequality (\ref{three}) in (\ref{LLL}), one obtains
\begin{equation}\label{l11}
	\mathcal{L}V_1(z_1,z_1')\leq (L_1+\frac{L_2 \varepsilon _1^{r_1}}{r_1}-\kappa_1)\Big((z_1-z_1')^T(z_1-z_1')\Big)^{\frac{k}{2}}+\frac{L_2}{s_1\varepsilon_1^{s_1}}\Big((z_2-z_2')^T(z_2-z_2')\Big)^{\frac{k}{2}}.
\end{equation} 
One can readily verify that $ V_1 $ is a $ \delta_\exists$-ISS-M$_k $ function for $ \zeta_1 $-subsystem with respect to $ z_2 $ as the input provided that $L_1+\frac{L_2 \varepsilon _1^{r_1}}{r_1}-\kappa_1<0$. Function $V_1$ satisfies the conditions in Definition \ref{definition3} with $ \underline{\alpha}(y)=\overline{\alpha}(y)=\frac{1}{k}y $, \textbf{d} is the natural Euclidean matric, $ \kappa=k(\kappa_1-L_1-\frac{L_2 \varepsilon _1^{r_1}}{r_1}) $, and $ \varphi(y)=\frac{L_2}{s_1\varepsilon_1^{s_1}}y $, for any $ y\in\R^+_0 $.\\
Now consider a Lyapunov function $ V_2(z_2,z_2') $, $ \forall z_2,z_2' \in \mathbb{R}^n $,  for the $ \zeta_2 $-subsystem as 
\begin{equation}\nonumber
	V_2(z_2,z_2')=\frac{1}{k}\Big((z_2-z_2')^T(z_2-z_2')\Big)^{\frac{k}{2}}.
	\label{Lyap_2}
\end{equation} 
The respective infinitesimal generator is given by 
\begin{align}\notag
	\mathcal{L}V_2(z_2&,z_2')=(z_2-z_2')^T\Big((z_2-z_2')^T(z_2-z_2')\Big)^{\frac{k}{2}-1}\\\notag
	&\Big[\big(Gu+\eta+\kappa_1 \frac{\diff M}{\diff t} z_1+\kappa_1(z_2-\kappa_1Mz_1)\big)\hspace{-.2em}-\hspace{-.2em}\big(G'u'+\eta'+\kappa_1 \frac{\diff M'}{\diff t} z_1'+\kappa_1(z_2'-\kappa_1M'z_1')\big)\Big]\\\notag
	&+\frac{1}{2} \Tr\Big(\big(\sigma(z_1)-\sigma(z_1')\big)\big(\sigma(z_1)-\sigma(z_1')\big)^T\partial_{z_2z_2}V_2(z_2,z_2')\Big)\\\notag
	&+\frac{1}{k}\sum_{i=1}^{\tilde{r}}\lambda_i\Big(\big((z_2+\rho(z_1)e_i)-(z_2'+\rho(z_1')e_i)\big)^T\big((z_2+\rho(z_1)e_i)-(z_2'+\rho(z_1')e_i)\big)\Big)^{\frac{k}{2}}\\\label{l2}
	&-\frac{1}{k}\sum_{i=1}^{\tilde{r}}\lambda_i\Big((z_2-z_2')^T(z_2-z_2')\Big)^{\frac{k}{2}},
\end{align}
where $ G=G(z_1) $, $ \eta=\eta(z_1,z_2+\alpha(z_1)) $, $ M=M(z_1)$, $G'=G(z_1') $, $ \eta'=\eta(z_1',z_2'+\alpha(z_1')) $, and $ M'=M(z_1') $. The same abbreviation will be used in the rest of the proof. The first term can be simply handled by selecting proper control input $ u $ and the second term can be reduced using consistency of norm, Lipschitz assumption on the diffusion term $ \sigma(\cdot)  $ and the Young's inequality as follows
\begin{align}
	&\frac{1}{2} \Tr\Big(\big(\sigma(z_1)-\sigma(z_1')\big)\big(\sigma(z_1)-\sigma(z_1')\big)^T\partial_{z_2z_2}V_2(z_2,z_2')\Big) \nonumber\\
	&=\frac{1}{2}\Tr\bigg(\big(\sigma(z_1)-\sigma(z_1')\big)\big(\sigma(z_1)-\sigma(z_1')\big)^T\Big[\Big((z_2-z_2')^T(z_2-z_2')\Big)^{\frac{k}{2}-1}I_n \nonumber\\ 
	&\ \ \ \ +(k-2)(z_2-z_2')(z_2-z_2')^T\Big((z_2-z_2')^T(z_2-z_2')\Big)^{\frac{k}{2}-2}\Big]\bigg)\nonumber\\
	&\leq\frac{k-1}{2}\|\sigma(z_1)-\sigma(z_1')\|_F^2\|z_2-z_2'\|^{k-2} \leq\frac{\min\{n,r\}(k-1)}{2}\|\sigma(z_1)-\sigma(z_1')\|^2\|z_2-z_2'\|^{k-2} \nonumber\\
	&\leq\frac{\min\{n,r\}L_\sigma^2(k-1)}{2}\|z_1-z_1'\|^2\|z_2-z_2'\|^{k-2} \nonumber\\
	&\leq\frac{\min\{n,r\}L_\sigma^2(k-1)}{2}\Big[\frac{\varepsilon _2^{r_2}}{r_2}\Big(( z_1-  z_1')^T(z_1-z_1')\Big)^{\frac{k}{2}}+\frac{1}{s_2\varepsilon_2^{s_2}}\Big((z_2-z_2')^T(z_2-z_2')\Big)^{\frac{k}{2}}\Big] \label{trace},
\end{align}

where
$s_2=\frac{k}{k-2}$, $r_2=\frac{k}{2},$ and $ \varepsilon_2 $ is any positive constant. 
With the help of Jenson's inequality for convex functions \cite{aujla2003weak} and of Lipschitz assumption on the reset term $ \rho(\cdot) $ (cf. Definition \ref{definition1}), the third term in (\ref{l2}) can be reduced as
\begin{align}
	\frac{1}{k}\sum_{i=1}^{\tilde{r}}\lambda_i&\Big[\|(z_2-z_2')+\big(\rho(z_1)e_i-\rho(z_1')e_i\big)\|^{k}-\Big((z_2-z_2')^T(z_2-z_2')\Big)^{\frac{k}{2}}\Big] \nonumber\\
	&\leq \frac{1}{k}\sum_{i=1}^{\tilde{r}}\lambda_i\Big[2^{k-1}\|z_2-z_2'\|^k+2^{k-1}L_\rho^k \|z_1-z_1'\|^{k}-\Big((z_2-z_2')^T(z_2-z_2')\Big)^{\frac{k}{2}}\Big]\nonumber\\
	&\leq \Big((z_2-z_2')^T(z_2-z_2')\Big)^{\frac{k}{2}}\frac{(2^{k-1}-1)\lambda}{k}+\Big((z_1-z_1')^T(z_1-z_1')\Big)^{\frac{k}{2}}\frac{2^{k-1}L_\rho^k \lambda}{k}\label{Poisson},
\end{align}
where $ \lambda=\sum_{i=1}^{\tilde{r}} \lambda_i$. Finally, the infinitesimal generator (\ref{l2}) corresponding to $V_2(z_2,z_2')$ can be reduced with the help of (\ref{trace}) and (\ref{Poisson}) to
\begin{align}	\notag
	\mathcal{L}&V_2(z_2,z_2')
	\leq\Big(( z_1-  z_1')^T(z_1-z_1')\Big)^{\frac{k}{2}}\Big(\frac{\min\{n,r\}L_\sigma^2\varepsilon _2^{r_2}(k-1)}{2r_2}+\frac{2^{k-1}L_\rho^k \lambda}{k}\Big)\\\notag
	&+(z_2-z_2')^T\Big((z_2-z_2')^T(z_2-z_2')\Big)^{\frac{k}{2}-1}\\\notag
	&\ \Big[\Big(Gu+\eta+\kappa_1 \frac{\diff M}{\diff t} z_1+\kappa_1(z_2-\kappa_1Mz_1)+\Big(\frac{(2^{k-1}-1)\lambda}{k}+\frac{\min\{n,r\}L_\sigma^2(k-1)}{2s_2\varepsilon_2^{s_2}}\Big)z_2\Big)\\
	&\ -\Big(G'u'+\eta'+\kappa_1 \frac{\diff M'}{\diff t}z_1'+\kappa_1(z_2'-\kappa_1M'z_1')+\Big(\frac{(2^{k-1}-1)\lambda}{k}+\frac{\min\{n,r\}L_\sigma^2(k-1)}{2s_2\varepsilon_2^{s_2}}\Big)z_2'\Big)\Big].\label{l3}
\end{align} 
Now consider the Lyapunov function $ V $ for the overall system (\ref{syst1}) as $ V(z,z')=V_1(z_1,z_1')+V_2(z_2,z_2') $ and the respective infinitesimal generator can be obtained by using (\ref{l11}) and (\ref{l3}) as
\begin{equation}
	\begin{split}
		&\mathcal{L}V(z,z')\leq\Big(L_1+\frac{L_2 \varepsilon _1^{r_1}}{r_1}+\frac{\min\{n,r\}L_\sigma^2\varepsilon _2^{r_2}(k-1)}{2r_2}+\frac{2^{k-1}L_\rho^k \lambda}{k}-\kappa_1\Big)\Big((z_1-z_1')^T(z_1-z_1')\Big)^{\frac{k}{2}}\\
		&+(z_2-z_2')^T\Big((z_2-z_2')^T(z_2-z_2')\Big)^{\frac{k}{2}-1}\\
		&\Big[\Big(Gu+\eta+\kappa_1 \frac{\diff M}{\diff t} z_1+\kappa_1(z_2-\kappa_1Mz_1)\hspace{-.2em}+\hspace{-.2em}\Big(\frac{(2^{k-1}-1)\lambda}{k}+\frac{L_2}{s_1\varepsilon_1^{s_1}}+\frac{\min\{n,r\}L_\sigma^2(k-1)}{2s_2\varepsilon_2^{s_2}}\Big)z_2\Big)-\\
		&\Big(G'u'+\eta'+\kappa_1 \frac{\diff M'}{\diff t}z_1'+\kappa_1(z_2'-\kappa_1M'z_1')\hspace{-.2em}+\hspace{-.2em}\Big(\frac{(2^{k-1}-1)\lambda}{k}\hspace{-.2em}+\hspace{-.2em}\frac{L_2}{s_1\varepsilon_1^{s_1}}\hspace{-.2em}+\hspace{-.2em}\frac{\min\{n,r\}L_\sigma^2(k-1)}{2s_2\varepsilon_2^{s_2}}\Big)z_2'\Big)\Big].\\
	\end{split}
	\label{l4}
\end{equation}
If we choose the state feedback control law $ u(z_1,z_2) $ as
\begin{equation}
	\begin{split}
		u(z_1,z_2)=& G^{-1}\Big(-\eta-\kappa_1 \frac{\diff M}{\diff t} z_1+\kappa_1^2Mz_1\\&-\Big(2\kappa_1+\frac{(2^{k-1}-1)\lambda}{k}+\frac{L_2}{s_1\varepsilon_1^{s_1}}+\frac{\min\{n,r\}L_\sigma^2(k-1)}{2s_2\varepsilon_2^{s_2}}\Big)z_2+\hat{u}\Big),
		\nonumber
	\end{split}
	\label{control_u}
\end{equation}
where $ \hat u $ is a new control input with respect to which the closed-loop system will be shown to be $ \delta_\exists$-ISS-M$_k$. After using $ u(z_1,z_2) $, the inequality (\ref{l4}) reduces to
\begin{align}\notag
	\mathcal{L}&V(z,z')\leq \hspace{-.2em}-\Big(\kappa_1\hspace{-.2em}-\big(L_1+\frac{L_2 \varepsilon _1^{r_1}}{r_1}\hspace{-.2em}+\hspace{-.2em}\frac{\min\{n,r\}L_\sigma^2\varepsilon _2^{r_2}(k-1)}{2r_2}\hspace{-.2em}+\hspace{-.2em}\frac{2^{k-1}L_\rho^k \lambda}{k}\big)\Big)\Big((z_1\hspace{-.2em}-\hspace{-.2em}z_1')^T(z_1\hspace{-.2em}-\hspace{-.2em}z_1')\Big)^{\frac{k}{2}}\\\label{asdf}
	&-\kappa_1\Big((z_2-z_2')^T(z_2-z_2')\Big)^{\frac{k}{2}}
	+(z_2-z_2')^T\Big((z_2-z_2')^T(z_2-z_2')\Big)^{\frac{k}{2}-1}(\hat{u}-\hat{u}').
\end{align}
Now the third term can further be reduced by applying Young's inequality to 
\begin{equation}
	\begin{split}
		(z_2-z_2')^T\Big((z_2-z_2')^T(z_2-z_2')&\Big)^{\frac{k}{2}-1}(\hat{u}-\hat{u}')
		\leq\|z_2-z_2\|^{k-1}\|\hat{u}-\hat{u}'\|\\
		&\leq\frac{\varepsilon _1^{r_1}}{r_1}\Big(( z_2-  z_2')^T(z_2-z_2')\Big)^{\frac{k}{2}}+\frac{1}{s_1\varepsilon_1^{s_1}}\|\hat{u}-\hat{u}'\|^k,
	\end{split}
	\label{yong_2}
\end{equation}
where the parameters $ \varepsilon_1, s_1 $ and $ r_1 $ are the same as the ones in (\ref{three}). Using \eqref{yong_2}, inequality (\ref{asdf}) reduces to
\begin{equation}
	\begin{split}
		\mathcal{L}V(z,z')\leq& -c_1\Big((z_1-z_1')^T(z_1-z_1')\Big)^{\frac{k}{2}}-c_2\Big((z_2-z_2')^T(z_2-z_2')\Big)^{\frac{k}{2}}+c_3\|\hat{u}-\hat{u}'\|^k,
	\end{split}
	\label{asd}
\end{equation}
where $c_1=\Big(\kappa_1-\big(L_1+\frac{L_2 \varepsilon _1^{r_1}}{r_1}+\frac{\min\{n,r\}L_\sigma^2\varepsilon _2^{r_2}(k-1)}{2r_2}+\frac{2^{k-1}L_\rho^k \lambda}{k}\big)\Big)$, $c_2=\Big(\kappa_1-\frac{\varepsilon _1^{r_1}}{r_1}\Big)$, $ c_3=\frac{1}{s_1\varepsilon_1^{s_1}} $, all required to be positive. By choosing the design parameter $ \kappa_1 $ as 
\begin{equation}
	\kappa_1> L_1+\frac{\max\{L_2,1\} \varepsilon _1^{r_1}}{r_1}+\frac{\min\{n,r\}L_\sigma^2\varepsilon _2^{r_2}(k-1)}{2r_2}+\frac{2^{k-1}L_\rho^k \lambda}{k},
	\label{condition1}
	\nonumber
\end{equation} 
one obtains $ c_1 $, $ c_2 $, $ c_3>0 $.
\\If $ \kappa=\min\{kc_1,kc_2\} $, the inequality (\ref{asd}) can further be reduced to  
\begin{equation}
	\begin{split}
		\mathcal{L}V\leq& -\kappa V(z,z')+\varphi(\|\hat{u}-\hat{u}'\|^k),
	\end{split}
	\label{pqrs0}
\end{equation}
where $ \varphi(y)=c_3y, \forall y \in \mathbb{R}_0^+ $, which satisfies condition (iii) of Definition \ref{definition3}. One can readily verify that conditions (i) and (ii) in Definition \ref{definition3} are satisfied by defining metric \textbf{d} as the natural Euclidean one, and defining $ \underline{\alpha}(y)=\dfrac{1}{2^{\frac{k}{2}-1}k}y $, and $ \overline{\alpha}(y)=\dfrac{1}{k}y, \forall y\in\mathbb{R}_0^+ $. Now with the help of Theorem \ref{theorem1}, one obtains
\begin{equation}
	\begin{split}
		\EE[\|\zeta_{z\hat{\upsilon}}(t)&-\zeta_{z'\hat{\upsilon}'}(t)\|^k]\leq \beta(\EE[\|z-z'\|^k],t)+\gamma(\EE[\|\hat{\upsilon}-\hat{\upsilon}'\|^k_\infty]),
	\end{split}
	\label{pqrs}
\end{equation}
where $ \zeta_{z\hat{\upsilon}}(t) $ denotes the value of the solution process of $ \hat{\Sigma} $ in (\ref{syst1})  at time $ t\in \mathbb{R}_0^+ $ under the input signal $ \hat{\upsilon} $ and from the initial condition $ \zeta_{z\hat{\upsilon}}(0)=z $ $ \PP $-a.s. The $ \mathcal{KL} $ function $ \beta $, and the $ \mathcal{K}_\infty $ function $ \gamma $ can be defined as
\begin{equation}
	\begin{split}
		\beta(y,t)&=\underline{\alpha}^{-1}(\overline{\alpha}(y) \mathsf{e}^{-\kappa t})=2^{\frac{k}{2}-1}\mathsf{e}^{-\kappa t}y,~~\gamma(y)=\underline{\alpha}^{-1}(\frac{\varphi(y)}{\kappa})=\frac{2^{\frac{k}{2}-1}k}{\kappa}c_3y,
		\label{pqrs1}
	\end{split}
\end{equation}
for all $ y \in \mathbb{R}_0^+ $. Now by applying the change of coordinate $ \zeta=\psi(\xi) $, where $ \xi=[\q^T~\p^T]^T $, the control law $\upsilon $ reduces to 
\begin{equation}
	\begin{split}
		\upsilon= G^{-1}(\q)\Big(&-\eta(\q,\p)+\kappa_1 \frac{\diff M(\q)}{\diff t} \q-\kappa_1^2M(\q)\q\\
		&-\Big(\frac{(2^{k-1}-1)\lambda}{k}+\frac{L_2}{s_1\varepsilon_1^{s_1}}+\frac{\min\{n,r\}L_\sigma^2(k-1)}{2s_2\varepsilon_2^{s_2}}\Big)\big(\p+\kappa_1M(\q)\q\big)+\hat{\upsilon}\Big),
	\end{split}
	\label{control_u1}
\end{equation}
and (\ref{pqrs}) can be rewritten as
\begin{equation}
	\begin{split}
		\EE[\|&\psi(\xi_{x\hat{\upsilon}}(t))-\psi(\xi_{x'\hat{\upsilon}'}(t))\|^k] \leq \beta(\EE[\|\psi(x)-\psi(x')\|^k],t)+\gamma(\EE[\|\hat{\upsilon}-\hat{\upsilon}'\|^k_\infty]),
	\end{split}
	\label{abc1}
\end{equation}
where $ x=[q^T~p^T]^T $. By defining a metric\footnote{Since $\psi$ is a bijective function, $\textbf{d}$ satisfies all the requirements of a metric.} $ \textbf{d}(x,x')=\|\psi(x)-\psi(x')\| $, we can rewrite (\ref{abc1}) as
\begin{equation}
	\begin{split}
		\EE[(\textbf{d}(\xi_{x\hat{\upsilon}}(t), &\xi_{x'\hat{\upsilon}'}(t)))^k] \leq \beta(\EE[(\textbf{d}(x,x'))^k],t)+\gamma(\EE[\|\hat{\upsilon}-\hat{\upsilon}'\|^k_\infty]),
	\end{split}
	\label{prop}
\end{equation}
which satisfies condition (\ref{nnn}) for original $ \Sigma $. Hence, $\Sigma$ in (\ref{ham2}) equipped with the feedback control law (\ref{control_u1}) is $ \delta_\exists$-ISS-M$_k$ for any $ k>2 $.
\end{proof}
The next corollary provides the same results as the ones in Theorem \ref{theorem2} but for $ k=2 $.

\begin{corollary}
	Consider the stochastic control system $ \Sigma $ in (\ref{ham2}). The state feedback control law 
	\begin{equation}
		\upsilon= G^{-1}(\q)\Big(-\eta(\q,\p)-\kappa_1 \frac{\diff M(\q)}{\diff t} \q+\kappa_1^2M(\q)\q-\Big(2\kappa_1+\frac{\lambda}{2}+\frac{L_2}{2\varepsilon_1^{2}}\Big)\big(\p+\kappa_1M(\q)\q\big)+\hat{\upsilon}\Big), \nonumber
	\end{equation}
	renders the closed-loop stochastic control system $\delta_\exists$-ISS-$ M_2$ with respect to input $ \hat{\upsilon} $, for all 
	\[ \kappa_1> L_1+\frac{\max\{L_2,1\} \varepsilon _1^{2}}{2}+\frac{\min\{n,r\}L_\sigma^{2}}{2}+L_\rho^2 \lambda ,\] 
	where $ \varepsilon_1 $ is any positive constant which can be chosen arbitrarily, and $ L_1, L_2, L_\sigma$, and $L_\rho $ are the Lipschitz constants introduced in (\ref{condition}) and Definition \ref{definition1}, respectively.
\end{corollary}

\begin{proof}
	The corollary is a particular case of Theorem \ref{theorem2}. The proof is almost similar to that of Theorem \ref{theorem2} by substituting $ k=2 $. The only difference will appear while handling the trace term (\ref{trace}) in $ \zeta_2 $-subsystem which is now given by
	\begin{equation}
		\begin{split}
			\frac{1}{2} &\Tr\Big(\big(\sigma(z_1)-\sigma(z_1')\big)\big(\sigma(z_1)-\sigma(z_1')\big)^T\partial_{z_2z_2}V_2(z_2,z_2')\Big)\\&\leq \frac{1}{2} \Tr\Big(\big(\sigma(z_1)-\sigma(z_1')\big)\big(\sigma(z_1)-\sigma(z_1')\big)^T\Big)\leq \frac{\min\{n,r\}L_\sigma^2}{2}(z_1-z_1')^T(z_1-z_1').
			\nonumber
		\end{split}
	\end{equation} 
	The rest of the proof follows similarly the one of Theorem \ref{theorem1}.
\end{proof}

\begin{remark} \label{remark1}
	Assume that for all $ x, x'\in\R^n $, the change of coordinate map $\psi$ in (\ref{cordi_tran}) satisfies 
	\begin{equation}
		\underline{\chi}(\|x-x'\|^k)\leq\|\psi(x)-\psi(x')\|^k\leq \overline{\chi}(\|x-x'\|^k),
		\label{new}
		\nonumber
	\end{equation}
	for some $ \mathcal{K}_\infty $ convex function $ \underline{\chi} $ and $ \mathcal{K}_\infty $ concave function $ \overline{\chi} $. Then, inequality (\ref{prop}) for the original system $ \Sigma $ reduces to
	\[\EE[\|\xi_{x\hat{\upsilon}}(t)-\xi_{x'\hat{\upsilon}'}(t)\|^k]\leq \hat{\beta}(\EE[\|x-x'\|^k],t)+\hat{\gamma}(\EE[\|\hat{\upsilon}-\hat{\upsilon}'\|^k_\infty]),\]
	for the $ \mathcal{KL} $ function $ \hat{\beta}(y,t)= \underline{\chi}^{-1}(2\beta(\overline{\chi}(y),t))$ and the $ \mathcal{K}_\infty $ function $ \hat{\gamma}(y)= \underline{\chi}^{-1}(2\gamma(y))$, for any $y,t\in\R^+_0 $. Note that if $\ul\chi^{-1}$ satisfies the triangle inequality (i.e., $\ul\chi^{-1}(a+b)\leq\ul\chi^{-1}(a)+\ul\chi^{-1}(b)$), one can remove the coefficients $2$ in the expressions of $\hat\beta$ and $\hat\gamma$.
	Particularly, if the inertia matrix ($ M $) is constant, one has 
	\begin{equation}
		\begin{split}
			\left \| 
			\begin{bmatrix}
				q-q' \\	(p+\kappa_1Mq)-(p'+\kappa_1Mq') \end{bmatrix} \right \|
			= \left \|A \begin{bmatrix}q-q'\\ p-p'\end{bmatrix} \right \|=\|A(x-x')\|,
		\end{split}
		\nonumber
	\end{equation}
	where $A$ is a constant matrix given by
	\[A=\begin{bmatrix}
	I_n & 0_n\\ 
	\kappa_1M& I_n
	\end{bmatrix}.\]
	Therefore, one obtains
	\[(\lambda_{\min}(A^TA))^{\frac{k}{2}}\|x-x'\|^k\leq\|\psi(x)-\psi(x')\|^k=\|A(x-x')\|^k\leq(\lambda_{\max}(A^TA))^{\frac{k}{2}}\|x-x'\|^k,\]
	where $ \lambda_{\min}(A^TA) $ and $ \lambda_{\max}(A^TA) $ denote minimum and maximum eigenvalues of $ A^TA $, respectively.
\end{remark} 

\section{Case Study}\label{IV}
\begin{figure}
	\centering
	\includegraphics[scale=.25]{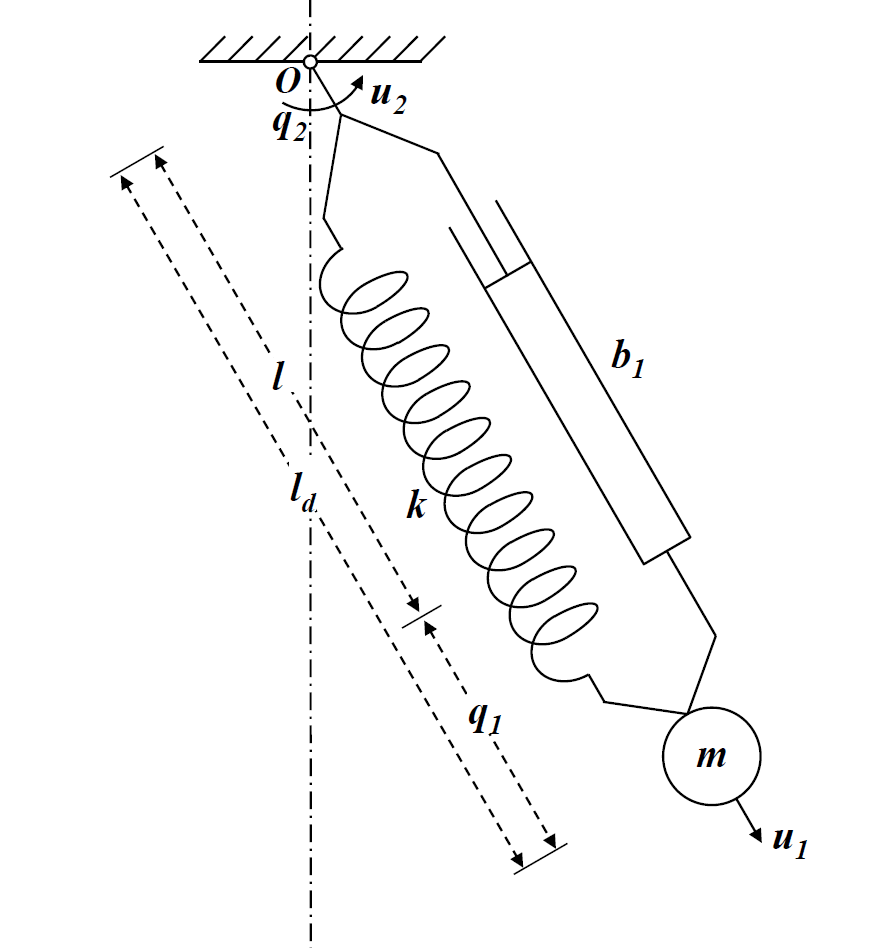}
	\caption{Controlled spring pendulum}
	\label{fig1}
	\vspace{-5mm}
\end{figure}
To verify the efficacy of the control design framework proposed in this paper, we illustrate the results on a spring pendulum attached to stochastically vibrating ceiling and subject to random  jumps such as sudden jerks due to interaction with environmental disturbances. The nonlinear dynamics of the considered system is borrowed from \cite{wu2012backstepping}, now affected by jumps and schematically shown in Figure \ref{fig1}. Let us define the generalized coordinate vector as $\q =[\q_1~ \q_2]^T$, where $ \q_1 $ represents change of arm length as a difference between the dynamic length ($ l_d $) and static length ($ l $) of a spring pendulum; and $ \q_2 $ is the angle of pendulum with vertical axis. The corresponding generalized momenta vector is given by $ \p =[m\frac{\diff \q_1}{\diff t}~~ m(l+\q_1)^2\frac{\diff \q_2}{\diff t}]^T$, where $ m $ is the mass of the ball, which gives the inertia matrix $ M(q) $ as
\begin{equation}
	M(q)=\begin{bmatrix}
		m & 0\\ 
		0 & m(l+q_1)^2
	\end{bmatrix}.
	\label{mass}
\end{equation}
\begin{figure}[t]
	\centering
	\includegraphics[scale=0.7]{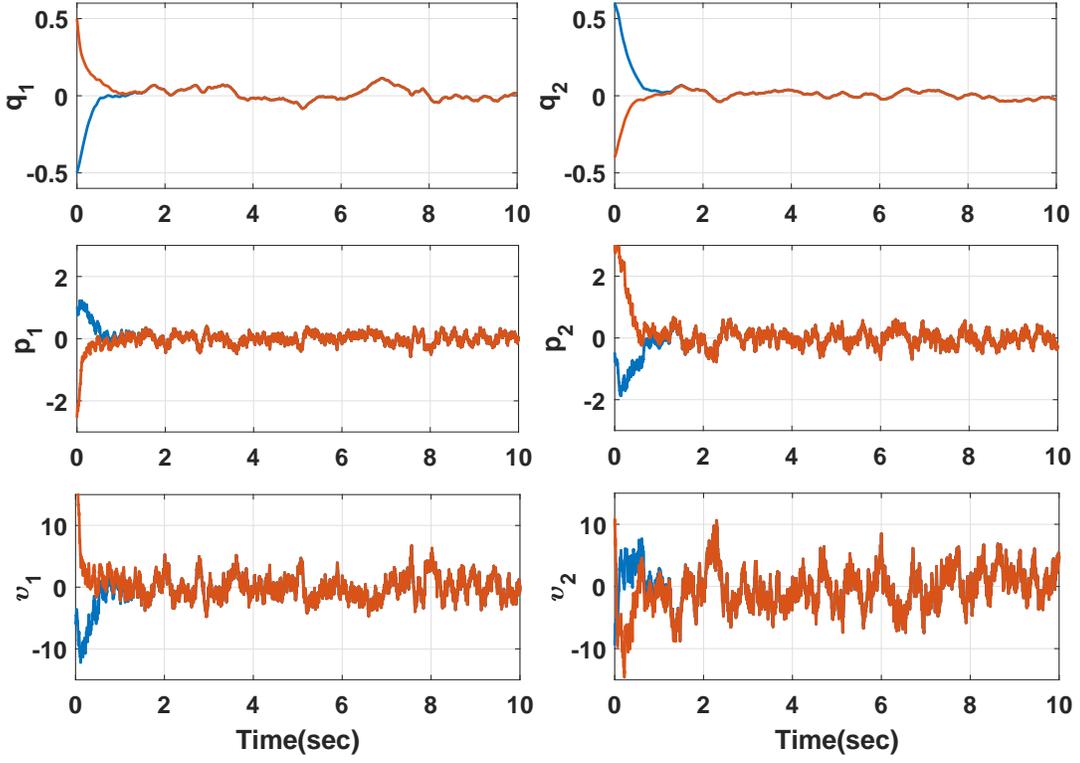}
	\caption{Two trajectories $ \q $ (top two plots), two trajectories $\p$ (middle two plots) started from two different initial conditions $ [q; p] =[0.5;  - 0.4;  - 2.5; 3]$ and $ [q'; p'] =[ - 0.5; 0.6; 1;  - 0.5]$, and the two corresponding input trajectories $\upsilon_1$ and $\upsilon_2$ (bottom two plots)}.
	\label{fig2}
\end{figure}
The Hamiltonian function $ H(q,p) $ is given by the total energy of the system as 
\begin{equation}
	\begin{split}
		H(q,p)=\frac{p_1^2}{2m}&+\frac{p_2^2}{2m(l+q_1)^2}+\frac{k_sq_1^2}{2}+mg(l+q_1)(1-\cos q_2),
		\nonumber
	\end{split}
\end{equation}    
where $ k_s $ is an elasticity coefficient of spring and $ g $ is the acceleration due to gravity. Now $ \eta(q,p)=-\frac{\partial H}{\partial q}(q,p)+ b(q,p) $ can be calculated as
\begin{equation}
	\eta(q,p)=\begin{bmatrix}
		\frac{p_2^2}{m(l+q_1)^3}-k_sq_1-mg(1-\cos q_2)\\ 
		-mg(l+q_1) \sin q_2
	\end{bmatrix}+\begin{bmatrix}
	-\frac{b_1 p_1}{m}\\ 
	-\frac{b_2 p_2}{m}
\end{bmatrix},
\label{xi}
\end{equation}
where $ b_1 $ is a damping coefficient of piston and $ b_2 $ is an air damping coefficient. By considering a 2-dimensional Brownian motion, the diffusion function $ \sigma(q) $ can be determined with the help of notion of relative kinematics by considering point $ O $ in Figure \ref{fig1} stochastically vibrating \cite{wu2012backstepping} which is given by
\begin{equation}\notag
	\sigma(q)=\begin{bmatrix}
		-m \sin q_2 & m\cos q_2\\ 
		-m(l+q_1)\cos q_2&-m(l+q_1)\sin q_2 
	\end{bmatrix}. 
\end{equation}
To introduce abrupt jumps in the system, we consider a one dimensional Poison process with the rate $ \lambda=1 $ and linear reset function $ \rho(\q)=\q $. The term $ \frac{\diff M(\q)}{\diff t} $ can be obtained as
\begin{equation}
	\begin{split}
		\frac{\diff M(\q)}{\diff t}=\frac{\partial M(\q)}{\partial q^T}\times\Big(\frac{\diff \q}{\diff t} \otimes I_2\Big)=\begin{bmatrix}
			0 &0 \\ 
			0 &\frac{2(1+\q_1)\p_1}{m} 
		\end{bmatrix}.
	\end{split}
	\label{dmdt}
\end{equation}
As control input $ \upsilon=[\upsilon_1~\upsilon_2]^T $ itself acting on the mass, one gets $ G(q)=I_2 $.
\begin{figure}[t]
	\centering
	\includegraphics[scale=0.6]{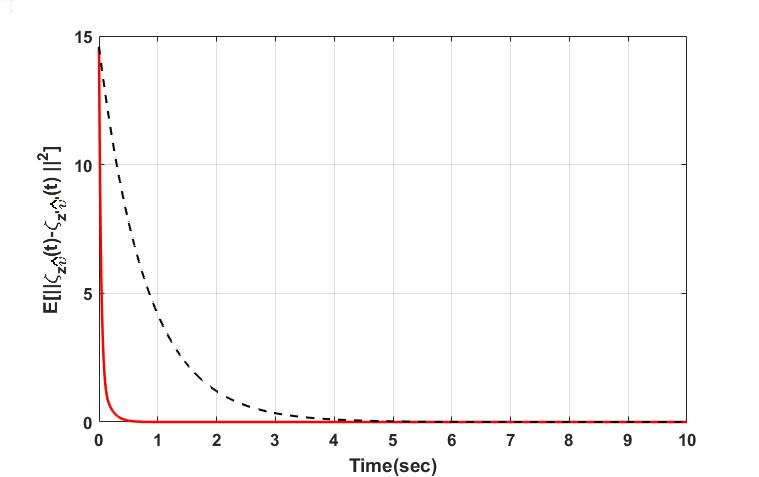}
	\caption{The average value of the squared distance of two trajectories of $ \hat{\Sigma} $ started from two different initial conditions $ z =[0.5; -0.4; -0.9; -2.12]$ and $ z' =[-0.5; 0.6; -0.6; 1.42]$. The black doted curve indicates corresponding bound given by (\ref{asdfg}).}
	\label{fig3}
	\vspace{-5mm}
\end{figure}
%
%
Now with the help of (\ref{mass}), (\ref{xi}), (\ref{dmdt}), Theorem \ref{theorem2}, and fixing $ k=2 $, we can obtain the final state feedback control input $ \upsilon $ for the considered system as follows
\begin{align}
	\upsilon_1(\q,\p)=&-\frac{\p_2^2}{m(l+\q_1)^3}+k_s\q_1+mg(1-\cos \q_2)+\frac{b_1\p_1}{m}-\kappa_1\Big(\kappa_1+\frac{\lambda}{2}+\frac{L_2}{2\varepsilon_1^{2}}\Big)m\q_1\nonumber\\
	&-\Big(2\kappa_1+\frac{\lambda}{2}+\frac{L_2}{2\varepsilon_1^{2}}\Big)\p_1+\hat{\upsilon}_1, \nonumber\\
	\upsilon_2(\q,\p)=&mg(l+\q_1) \sin \q_2+\frac{b_2\p_2}{m}+\frac{2\kappa_1(1+\q_1)\p_1\q_2}{m}-\kappa_1\Big(\kappa_1+\frac{\lambda}{2}+\frac{L_2}{2\varepsilon_1^{2}}\Big)m\q_2(l+\q_1)^2\nonumber\\
	&-\Big(2\kappa_1+\frac{\lambda}{2}+\frac{L_2}{2\varepsilon_1^{2}}\Big)\p_2+\hat{\upsilon}_2\nonumber,
\end{align}
rendering the closed-loop system $ \delta_\exists$-ISS-M$_2$ with respect to input $ [\hat{\upsilon}_1~\hat{\upsilon}_2]^T $ for any arbitrarily chosen $\varepsilon_1>0$ and appropriately chosen $\kappa_1$.
For the simulation purpose, we consider system parameters as $ m=0.8 $, $ l=1.5 $, $ g=9.8 $, $ k_s=15 $, $ b_1=1 $, and $ b_2=1 $; all the constants and the variables are considered in SI units; the Lipschitz constants are computed as $ L_1=1, L_2=2, L_\sigma=1$, and $L_\rho=1$, and the design parameters are chosen as $\varepsilon_1=0.5 $ and $\kappa_1=4$. We choose inputs $ \hat{\upsilon}_1(t) = \hat{\upsilon}_2(t) =0.5\sin t$. Figure \ref{fig2} shows the evolution of the closed-loop trajectories $ \q $ and $ \p $ in presence of Brownian noise and Poisson jumps started from two different initial conditions $ [q; p] $=[0.5; $ - $0.4; $ - $2.5; 3] and $ [q'; p'] $=[$ - $0.5; 0.6; 1; $ - $0.5] and the evolution of the corresponding input trajectories $\upsilon_1$ and $\upsilon_2$. Figure \ref{fig2} shows that indeed, by virtue of the $ \delta_\exists$-ISS-M$_2$ property, both trajectories converge to each other. To verify the bound on $\EE[\|\zeta_{z\hat{\upsilon}}(t)-\zeta_{z'\hat{\upsilon}'}(t)\|^2]$ as given in (\ref{pqrs}), we simulated the closed-loop system for 5000 realizations, two fixed initial conditions, and the same input for both trajectories (i.e $ \hat{\upsilon}=\hat{\upsilon}' $). The inequality (\ref{pqrs}) reduces to
\begin{equation}
	\EE[\|\zeta_{z\hat{\upsilon}}(t)-\zeta_{z'\hat{\upsilon}'}(t)\|^2]\leq\beta(\|z-z'\|^2,t),
	\label{asdfg}
\end{equation}
where the $ \mathcal{KL} $ function $ \beta $ is given in (\ref{pqrs1}) and computed as $\beta(y,t)=\mathsf{e}^{-\kappa t}y$ with $ \kappa =1.25$. The average value of the squared distance of two trajectories of $ \hat{\Sigma} $ started from two different initial conditions $ z =[0.5;\ -0.4;\ -0.9;\ -2.12]$ and $ z' =[-0.5;\ 0.6;\ -0.6;\ 1.42]$ together with computed theoretical bound are shown in Figure \ref{fig3}. One can readily verify that the simulated distance is always lower than the computed theoretical one in (\ref{asdfg}). 

\section{Conclusion}\label{V}
We introduced a coordinate invariant notion of incremental input-to-state stability for stochastic control systems with jumps and provided its description in terms of existence of a notion of so-called incremental Lyapunov functions. Furthermore, a backstepping controller design scheme was proposed for a class of nonlinear stochastic Hamiltonian systems with jumps. The design scheme provides controllers rendering the close-loop systems $\delta_\exists$-ISS-M$_k$. Finally, we illustrated the effectiveness of the results on a nonlinear stochastic Hamiltonian system.

\bibliographystyle{alpha}
\bibliography{IEEEtran1}
\end{document}